\documentclass[journal]{IEEEtran}
\usepackage{amsmath}
\usepackage{multicol}
\usepackage{mathtools}
\usepackage{amssymb}
\usepackage{amsmath}  
\usepackage{subeqnarray}
\usepackage{bm}
\usepackage{color}

\usepackage{algorithm}
\usepackage{algorithmic}        

\usepackage{cite}

\usepackage{array}  
\usepackage{url}

\usepackage{amsthm}
\usepackage{graphicx}
\usepackage{subfigure}
\theoremstyle{plain}     
\newtheorem{thm}{Theorem}
\newtheorem{lem}{Lemma}

\theoremstyle{definition}

\theoremstyle{remark}

\usepackage{epstopdf}

\begin{document}
%
\title{Game-Theoretic Optimization for Machine-Type Communications Under QoS Guarantee}

\author{\IEEEauthorblockN{Yu Gu,~\IEEEmembership{Student Member,~IEEE,}
                          Qimei Cui,~\IEEEmembership{Senior Member,~IEEE,}\\
                          Qiang Ye,~\IEEEmembership{Member,~IEEE,}
                          Weihua Zhuang,~\IEEEmembership{Fellow,~IEEE,}}
\thanks{Yu Gu and Qimei Cui are with the National Engineering Laboratory for Mobile Network Technologies, Beijing University of Posts and Telecommunications, Beijing 100876, China (e-mail: guyu@bupt.edu.cn; cuiqimei@bupt.edu.cn). Corresponding
author: Qimei Cui.}
\thanks{Qiang Ye and Weihua Zhuang are with the Department of Electrical and Computer Engineering, University of Waterloo, Waterloo, ON N2L 3G1, Canada (e-mail: q6ye@uwaterloo.ca; wzhuang@uwaterloo.ca).}
}

\maketitle

\begin{abstract}
Massive machine-type communication (mMTC) is a new focus of services in fifth generation (5G) communication networks. The associated stringent delay requirement of end-to-end (E2E) service deliveries poses technical challenges. In this paper, we propose a joint random access and data transmission protocol for mMTC to guarantee E2E service quality of different traffic types. First, we develop a priority-queueing-based access class barring (ACB) model and a novel effective capacity is derived. Then, we model the priority-queueing-based ACB policy as a non-cooperative game, where utility is defined as the difference between effective capacity and access penalty price. We prove the existence and uniqueness of Nash equilibrium (NE) of the non-cooperative game, which is also a sub-modular utility maximization problem and can be solved by a greedy updating algorithm with convergence to the unique NE. To further improve the efficiency, we present a price-update algorithm, which converges to a local optimum. Simulations demonstrate the performance of the derived effective capacity and the effectiveness of the proposed  algorithms.
\end{abstract}

\begin{IEEEkeywords}
Massive machine-type communications (mMTC), random access, effective capacity, quality of service (QoS), non-cooperative game, sub-modular function.
\end{IEEEkeywords}

%
\IEEEpeerreviewmaketitle

\section{Introduction}
With the rapid penetration of emerging automated services, machine-type communication (MTC) has become a focus of new services in fifth generation (5G) communication networks \cite{ye2017token}. For example, the smart cities span a wide variety of service applications, from traffic management to grid power distribution, to urban security and environmental monitoring \cite{2,21};  industrial wireless sensor devices bring the industry automation into the era of Industry 4.0 by providing the ubiquitous perception to improve the production efficiency \cite{3,new3}; health--oriented wearable devices offer biometric measurements such as heart rate, perspiration level, and oxygen level in the bloodstream \cite{31,new1}. Industry analysts predict that 50 billion devices will be interconnected via mobile networks worldwide by 2020 \cite{yaqoob2017internet}. Different from human-driven communication, MTC involves a massive number of devices with low-complexity and low-power, referred to as massive machine-type communications (mMTC). For mMTC, traffic overloading on the random access channel (RACH) is a dominate factor causing performance bottleneck due to simultaneous activation of devices. To alleviate the traffic arrival burstiness, the access class barring (ACB) \cite{6855701} and back-off procedure \cite{6720118} are proposed.

Another challenge is how to guarantee heterogeneous quality of-service (QoS), i.e., delay and reliability for devices with different traffic types \cite{4,fin}. For example, MTC for industrial automation requires extremely small delay and packer error probability (PER) \cite{7247338}. Although ACB or back-off procedure can prevent simultaneous massive access from activated devices, they result in a degradation of QoS. Thus, QoS-aware random access for mMTC has recently drawn attentions from researchers. An analytical model is proposed in \cite{8047949} to calculate maximum access delay for a distributed queueing-based random access protocol. In \cite{6423760}, the MTC network is modeled as a Beta/M/1 queueing system, where service time of the RACH procedure follows an exponential distribution. In \cite{7346396}, a general model is developed for the RACH procedure. The impact of short packets of machine-type services on the achievable rate is considered in \cite{7980731}. In traditional communications, the Shannon's channel coding theorem is employed under the assumption of infinity block-length, while for a finite block-length, a novel effective capacity is developed to guarantee statistical QoS for MTC. In \cite{8070468}, a cross-layer framework is proposed based on the effective bandwidth theory to optimize the packet dropping policy, power allocation policy, and bandwidth allocation policy under the QoS constraint. A probabilistic bound for RACH access delay is derived using stochastic network calculus in \cite{vilgelm2017reliability}.

The overall performance of MTC depends on the performance not only in the RACH phase, but also in data transmission phase. The backlog status of each transmission queue of MTC devices, depending on traffic arrival and service processes, affects the number of devices activated in both RACH phase and data transmission phase. As a result, a joint consideration of both RACH procedure and data transmission is a key to capture the performance of MTC. Ruan et al. in \cite{7996795} consider a scenario where devices can transmit data at a random access slot only when it decides to transmit a random preamble sequence and other devices have not occupied the preamble. The transmission process is formulated as an infinite horizon average cost Markov decision process (MDP). In \cite{6786066}, a queueing model is developed for MTC under the assumption that a device is activated when the number of packets in its queue is larger than a threshold and all the packets in the queue are transmitted upon one successful preamble contention. Wiriaatmadja et al. derive a closed-form formula for throughput as a function of the number of resource block allocated for the PRACH and the traffic load \cite{wiriaatmadja2015hybrid}. According to the closed-form formula, a hybrid protocol is proposed to maximize the expected throughput.

How MTC devices with differentiated services should contend to access the RACH and be scheduled to transmit data for heterogeneous QoS provisioning needs in-depth investigation \cite{new4,new5}. First, the RACH and data transmission phases are correlated, in presence of transmission queueing dynamics among devices. A unified analytical framework to quantify the QoS-provisioning capabilities in both random access and data transmission is essential. Moreover, how to guarantee differentiated QoS requirements of MTC devices supporting multiple traffic types is challenging. For massive access from MTC devices, a centralized control solution corresponds to a large communication overhead due to global information exchange with a large number of devices. Thus, a distributed QoS provisioning scheme is desired.

In this paper, we propose a joint random access and data transmission protocol to guarantee the end-to-end (E2E) QoS of different traffic types in mMTC. To deal with the aforementioned challenges, in this paper, we propose a hybrid random access and data transmission protocol to guarantee the end-to-end (E2E) QoS of different traffic types in mMTC. The main contributions of this paper can be summarized as follows:
\begin{itemize}
  \item First, we develop a priority-queueing-based ACB model and derive a novel effective capacity, taking into account the random access and short packet transmissions. The priority-queueing-based ACB model not only alleviates the traffic arrival burstiness, but also guarantees the differentiated QoS requirements of MTC devices.
  \item Second, based on the derived effective capacity, the barring policy is modeled as a non-cooperative game. We prove the existence and uniqueness of Nash equilibrium (NE) in the noncooperative game. Furthermore, a distributed iterative algorithm that converges to the equilibrium point is proposed.
  \item Thirdly, to improve the efficiency of NE, we provide a price-update algorithm, which converges to a locally optimal point of the effective capacity optimization problem. Simulation results show the performance of the novel effective capacity and the effectiveness of our proposed algorithms.
\end{itemize}

The remainder of this paper is organized as follows. The system model is presented in Section II. In Section III, we derive a new effective capacity, considering both RACH and data transmission phases. In Section IV, we present two game-theoretic barring policies with QoS guarantee. In Section V, simulation results are provided, followed by conclusions in Section VI.
\section{System Model}
\begin{figure}[!t]
\centering
\includegraphics[width=3in]{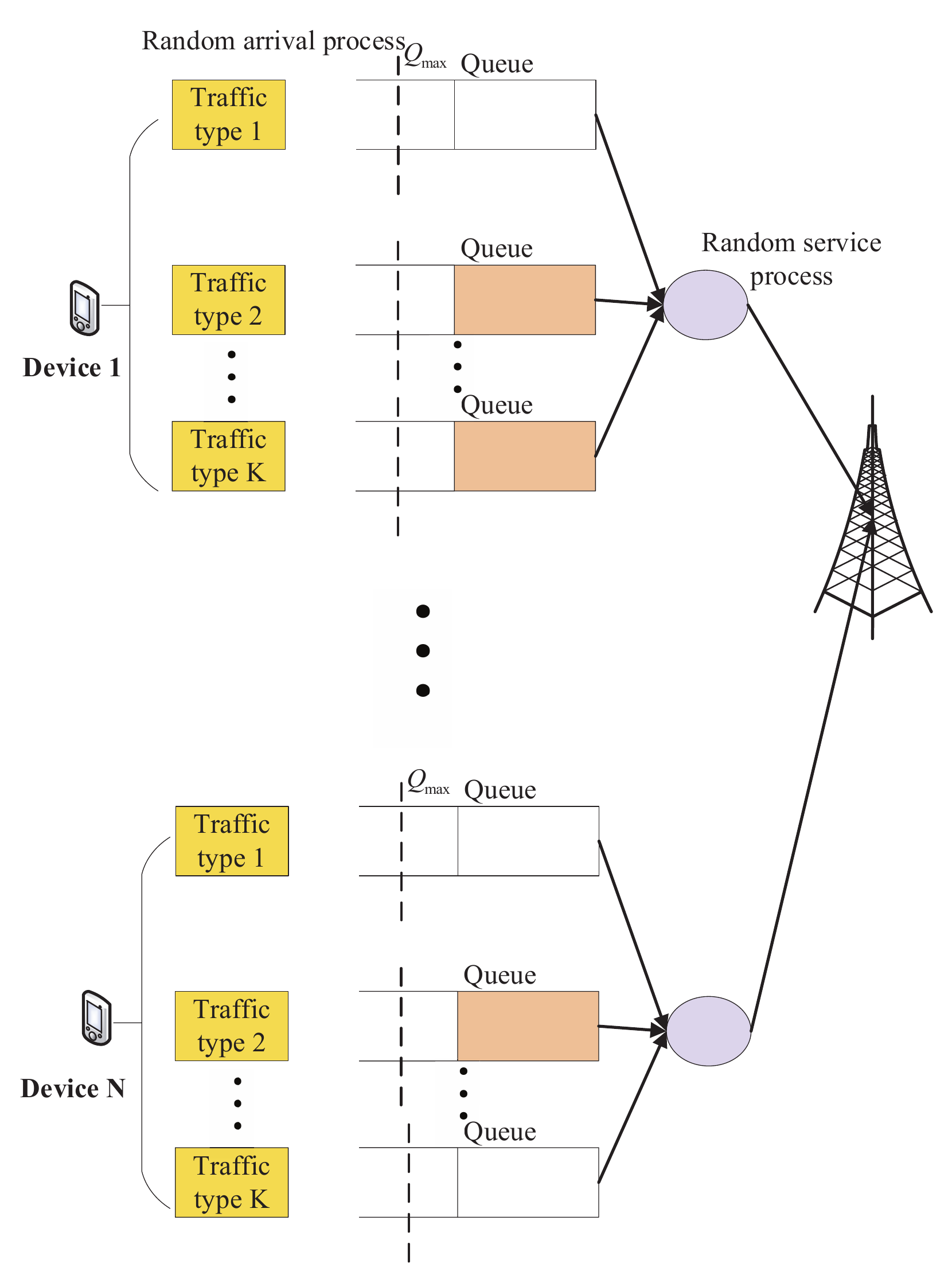}
\vspace{-0.4cm}
\caption{An example of system model including one BS and $N$ devices with three queue to transmit $K$ types of traffic .}
\label{fig1}
\vspace{-0.2cm}
\end{figure}

Consider the uplink of an mMTC system, composed of a BS and $N$ MTC devices which are randomly and uniformly distributed within the coverage of the BS, as shown in Fig. \ref{fig1}. Each device generates $K$ types of traffic, indexed by $1,2,...K$. We assume the smaller the index, the higher the transmission priority. Each device has $K$ First-Input-First-Output (FIFO) transmission queues for the $K$ types of traffic, respectively. Time is partitioned into a sequence of time slots, the duration of which is constant and denoted by $T_d$. Packet arrivals of the $k$-th traffic type at the $n$-th device over each time slot follows an independent and identically distributed Bernoulli process with probability $p_n$. The amount of the $k$-th traffic packet (in the unit of bit per second) at the $n$-th device at the $t$-th slot, denoted by $a_{n,k}[t]$, follows an exponential distribution with mean $\bar L_{n,k}$.



A conventional device initiates access attempts to the BS at the time when the transmission queues are non-empty. However, the simultaneous activation of a large number of devices causes access performance degradation for mMTC. To deal with massive access attempts with bursty traffic arrivals, a priority-queueing-based ACB is employed to reduce contention collisions of random access for QoS provisioning. Suppose the $k$-th queue of the $n$-th device is non-empty. If all the $k-1$ queues that have a lower index with higher priority than the $k$-th queue are empty, the $k$-th queue can be activated with probability $d_{n,k}$.

For uplink transmissions, time is partitioned into a sequence of transmission superframes, each of which consists of a \emph{random access phase} followed by a \emph{data transmission phase}, depicted in Fig. 2. The durations of a random access phase and a data transmission phase are denoted by $T_s$ and $T_f$, respectively, and the total transmission duration is denoted by $T_u$. The activated devices contend for preambles in $T_s$. Upon a successful contention, data is transmitted within a duration of $T_f$ using Orthogonal Frequency Division Multiple Access (OFDMA) techniques. Otherwise, the device will try again to access the BS in the random access phase of next superframe. We assume that ${T_s} \ll {T_f}$ to alleviate the signaling overhead due to random access. In general, the duration of each superframes $T_u$ is an integer multiple of time slots $T_d$. The wireless channel from a device to the BS is modeled as a Rayleigh block fading, i.e., the channel remains unchanged during a superframe of duration $T_u$, and changes independently between consecutive superframes.

In the \emph{random access phase}, a device attempts to connect to the BS based on a standard four step random access (RA) procedure. The detailed RA procedure for the uplink is as follows: 1) The device uniformly and randomly selects and then transmits one RA preamble from $M$ preambles in the physical random access channel (PRACH); 2) The BS sends random access responses (RARs) through the physical downlink shared channel (PDSCH) to devices whose preambles are decoded successfully; 3) Devices specified in the received RAR send the connection setup request messages; 4) If the BS correctly decodes a connection request, it replies with an acknowledgement to the device. In the \emph{data transmission phase}, we assume that the $k$-th queue of the $n$-th device successfully occupies a data transmission channel with the allocated bandwidth $B_n$ during the period $T_f$, which consists of $S_n$ OFDM symbols. We calculate $S_n$ by ${S_n} = \frac{{{T_f}}}{a}\frac{{{B_n}}}{c}$, where $a$ is the duration of a OFDM symbol, $c$ is the bandwidth of a OFDM symbol.

\begin{figure}[!t]
\centering
\includegraphics[width=3in]{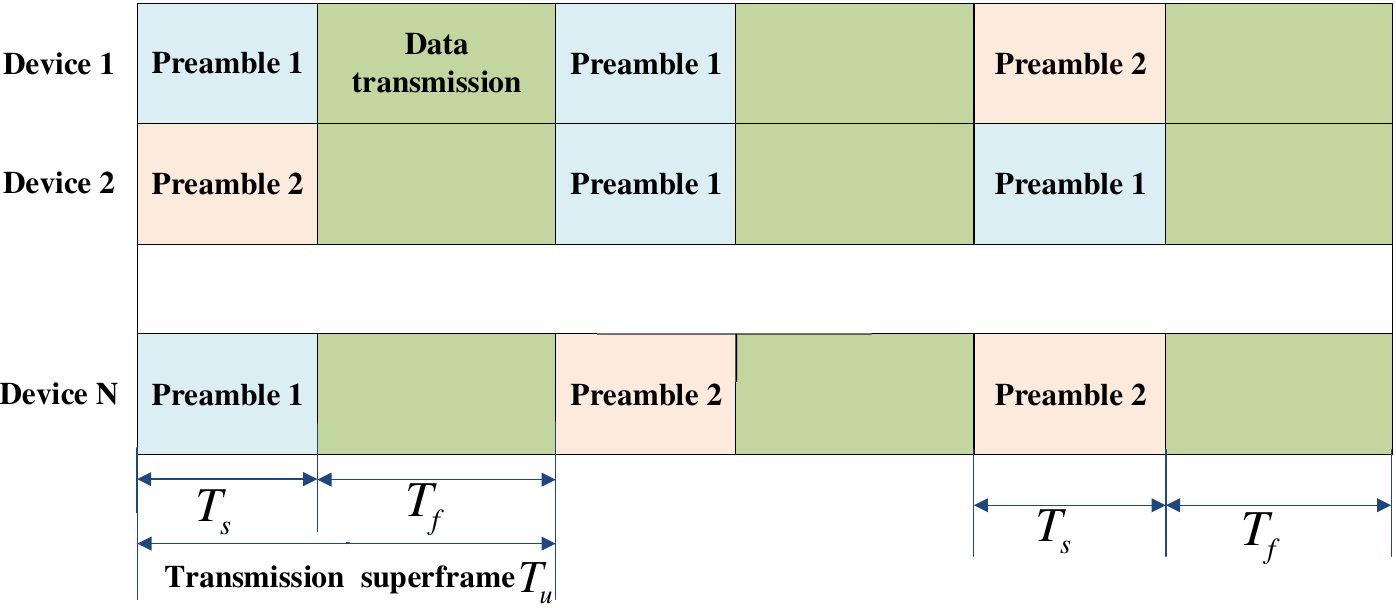}
\vspace{-0.4cm}
\caption{An illustration of MAC operations in each transmission superframe, including the random access phase with two preambles and the data transmission phase.}
\label{fig2}
\vspace{-0.2cm}
\end{figure}

\section{QoS Analysis of Multi-class MTC}
In this section, a novel analytical approach is proposed to analyze the performance of the priority-queueing-based ACB model, jointly considering the \emph{random access phase} and the \emph{data transmission phase}.

In the context of large deviation theory \cite{CuiEC}, QoS is characterized statistically by employing the QoS exponent $\bm \theta  = \left\{ {{\theta _{n,k}},n = 1, \cdots ,N}, k = 1, \cdots ,K  \right\}$, as given by \cite{EC}
\begin{equation}\label{1a1}
\theta_{n,k}= - \mathop {\lim }\limits_{{Q_{n,k}^{\rm th}} \to \infty } \frac{{\log (\Pr \{ Q_{n,k}[\infty] > {Q_{n,k}^{\rm th}}\} )}}{{{Q_{n,k}^{\rm th}}}}
\end{equation}
where $Q_{n,k}[t]$ is the length of transmission queue at the $k$-th queue of the $n$-th device at slot $t$, $Q_{n,k}^{\rm th}$ is the threshold of the queue length to guarantee the QoS of the $k$-th traffic type, and $\Pr \{ Q_{n,k}[\infty] > {Q_{n,k}^{\rm th}}\}$ is the QoS-violation probability that the queue length exceeds $Q_{n,k}^{\rm th}$. In (\ref{1a1}), $\theta_{n,k}$ provides the exponential decaying rate of the probability that the threshold is violated.
\subsection{Effective Bandwidth of Multi-class MTC}
For the arrival process $a_{n,k}[t]$ at the $k$-th queue of the $n$-th device, the effective bandwidth, denoted by $A_{n,k}(\theta_{n,k})$, specifies the minimum constant service rate for probabilistic QoS guarantee. Given the Bernoulli traffic arrival, the effective bandwidth (in bit/s) of the traffic is given by \cite{ChenEC}
\begin{equation}\label{EB}
\begin{aligned}
{A_{n,k}}({\theta _{n,k}}) &= \mathop {\lim }\limits_{t \to \infty } \frac{1}{{{\theta _{n,k}}t}}\log (\mathbb{E}\{ {e^{{\theta _{n,k}}\sum\limits_{i = 0}^t {{a_{n,k}}[i]} }}\} )\\
& = \frac{1}{{{\theta _{n,k}}{T_d}}}\log (\frac{{{p_n}}}{{1 - {\theta _{n,k}}{{\bar L}_{n,k}}}} + 1 - {p_n}).
\end{aligned}
\end{equation}
\subsection{Effective Capacity of Multi-class MTC}
Service for the $k$-th queue at the $n$-th device can have one of the following three states:
\begin{itemize}
  \item \textbf{Successful access state} accounts for the case that all queues of higher priority are empty, while the $k$-th queue is not empty and tries to access to the BS, and the chosen preamble is occupied only by the $n$-th device;
  \item \textbf{Collided access state} accounts for the case that all queues of higher priority are empty, while the $k$-th queue is not empty and tries to access to the BS, and the chosen preamble is occupied by more than one device;
  \item \textbf{Silence state} accounts for the case that more than one queue of higher priority is non-empty, and/or the $k$-th queue does not try to access to the BS.
\end{itemize}

We denote successful access indicator variable as:
\[{s_{n,k}} = \left\{ {\begin{array}{*{20}{l}}
{1,{\rm{if~the~state~is~successful~access~state}}};\\
{0,\rm{otherwise}}.
\end{array}} \right.\]

Let $P_{\text{idle}}^{_{n,k}}$ denote the probability of having an empty buffer at the $k$-th queue of $n$-th device. Then, the probability that the packet at the head of the $k$-th queue of the $n$-th device tries to access the BS is given by
\begin{equation}
P_a^{n,k} = {d_{n,k}}\prod\limits_{j = 1}^{k - 1} {P_{\text{idle}}^{_{n,j}}} (1 - P_{\text{idle}}^{_{n,k}}).
\end{equation}
The activation probability of $n$-th device is
\begin{equation}
{D_n} = \sum\limits_{k = 1}^K {P_a^{n,k}}.
\end{equation}
As a result, the successful access probability of the $k$-th queue of the $n$-th device is given by
\begin{equation} \label{12}
\begin{aligned}
{F_s^{n,k}} &= \sum\limits_{m \in M} {\frac{1}{M}} P_a^{n,k}\prod\limits_{l \ne n,l \in N} {\left( {1 - \frac{D_l}{M}} \right)} \\
 &= P_a^{n,k}\prod\limits_{l \ne n,l \in N} {\left( {1 - \frac{D_l}{M}} \right)}
\end{aligned}
\end{equation}
where $\left( {1 - \frac{D_l}{M}} \right)$ is the probability that the $l$-th device does not operate on the $m$-th preamble, and $M$ is the number of preambles.

For the short packet transmission, the achievable channel coding rate (bits per channel use) of the successful access state depends not only on the channel condition, but also on the finite blocklength $S_{n}$ and the packet error probability $\varepsilon_{n,k}$, which is given by \cite{7980731}
\begin{equation}\label{rate}
\begin{aligned}
{r_{n,k}}\approx & {\log _2}(1 + \frac{{P_{\text{th}}^{n,k}{G_n}{{\left| {{H_n}} \right|}^2}}}{{{N_0}}})\\
&- \!\!\sqrt {\frac{1}{{{S_n}}}\left( {1 - \frac{1}{{{{\left( {1 + \frac{{P_{\text{th}}^{n,k}{G_n}}}{{{N_0}}}{{\left| {{H_n}} \right|}^2}} \right)}^{\rm{2}}}}}} \right)} {Q^{ - 1}}(\varepsilon_{n,k} ){\log _2}(e).
\end{aligned}
\end{equation}
In (\ref{rate}), $Q(x) = \int_x^\infty  {\frac{1}{{\sqrt {2\pi } }}} {e^{\frac{{ - {t^2}}}{2}}}dt$ is the Gaussian Q-function, and ${Q^{ - 1}}(\cdot )$ represents its inverse, $N_0$ is the power of the additive white Gaussian noise, ${P_{\text{th}}^{n,k}}$ is the transmit power of the $k$-th queue of the $n$-th device, $G_n$ is the mean path loss attenuation coefficient (i.e., consisting of free-space loss and shadowing effect) from the $n$-th device to the BS, ${{H_{n}}}$ is the Rayleigh fading coefficient with unit variance, and ${{{\left| {{H_{n}}} \right|}^2}}, n=1,2,...N$ are independent and exponentially distributed random variables with unit mean, i.e., ${\left| {{H_{n}}} \right|^2} \sim\exp (1)$. From (\ref{rate}), we notice that there is a penalty on the maximum achievable rate since the finite blocklength makes the decoding error probability nonnegligible. Note that a packet consists of many payload bits, which are made up of information bits (information payload) and additional bits, containing metadata from the media access-control (MAC) layer and higher layers. The payload bits are typically encoded into a block of $S_n$ data symbols (complex numbers) to increase reliability in packet transmission.

Given the block Rayleigh fading during different transmission superframes, the effective capacity (in bit/s), denoted by $C_{n,k}({\theta _{n,k}},{d_{n,k}},{\bm d_{ - n}},P_{\text{th}}^{n,k})$, specifies the maximum, consistent, steady-state arrival rate at the input of the $k$-th queue of the $n$-th device, under the QoS constraint, given by \cite{EC,7980731}
\begin{equation} \label{EC}
\begin{aligned}
&{C_{n,k}}({\theta _{n,k}},{d_{n,k}},{\bm d_{ - n}},P_{\text{th}}^{n,k}) = \\
& - \frac{1}{{{\theta _{n,k}}{T_u}}}\log \left\{ {{F_s^{n,k}}(1 - {\varepsilon _n}){\mathbb{E}_{{H_n}}}\{ {e^{ - {\theta _{n,k}}{r_{n,k}}{S_n}}}\} } \right.\\
&\left. { + 1 - {F_s^{n,k}}(1 - {\varepsilon _n})} \right\}
\end{aligned}
\end{equation}
In Eq. (7), $\mathbb{E}_{H_n}\{\cdot\}$ is the expectation operation on small-scale channel fading for the $n$-th device, ${\bm{d}_{ - n}} = \left\{ {{{\bm d}_j},j \in \mathcal{N} - n} \right\}$, ${{\bm d}_n} = \left[ {{d_{n,1}},{d_{n,2}}, \ldots {d_{n,K}}} \right]$ is  barring strategy vector for device $n$, ${F_s^{n,k}}(1 - {\varepsilon _n})$ is the successful transmission probability, and $r_{n,k}$ is the achievable transmit rate, which is provided by Eq. (6). The proposed effective capacity is different from that of \cite{7980731}, considering the RACH procedure and ACB model.
\subsection{Composition Results}
If the assumptions of Gartner-Ellis theorem hold \footnote{The Gartner-Ellis theorem assumptions require the function (i) exists for all real $\theta$, (ii) is strictly convex, and (iii) is essentially smooth.} \cite{ChenEC,7061966}  and there is a unique QoS exponent $\theta^*_{n,k}$ that satisfies
\begin{equation} \label{Q1}
{A_{n,k}}({\theta^*_{n,k}}) = {C_{n,k}}({\theta^*_{n,k}},{d_{n,k}},{\bm d_{ - n}},P_{\text{th}}^{n,k})
\end{equation}
then the $k$-th queue of the $n$-th device is stable and $P({Q_{n,k}}[\infty ] \ge Q_{n,k}^{{\rm{th}}})$ can be approximated by \cite{ChenEC}
\begin{equation} \label{Q}
P({Q_{n,k}}[\infty ] \ge Q_{n,k}^{{\rm{th}}}) \approx (1-{P_{\text{idle}}^{_{n,k}}})e^{( - {\theta^*_{n,k}}Q_{n,k}^{{\rm{th}}})}.
\end{equation}
Note that Eq. (8) is the related ties between (2) and (7), which makes the effective capacity equal to the effective bandwidth, in order to satisfy the QoS.

Given a delay bound, $D_{\max}^{n,k}$, the probability that the steady-state packet delay at the $k$-th queue of the $n$-th device exceeds $D_{\max}^{n,k}$ is given by \cite{7061966}
\begin{equation}
{P^{n,k}_{th}} \approx  (1-{P_{\text{idle}}^{_{n,k}}}){e^{ - {\theta^*_{n,k}}{A_{n,k}}({\theta^*_{n,k} }){D^{n,k}_{\max }}}}.
\end{equation}
Given the Bernoulli traffic arrivals, $P_{\text{idle}}^{_{n,k}}$ is approximated as $\theta _{n,k}^*\bar L_{n,k}$ \cite{ChenEC}. The empty buffer probability of the queue can be also estimated by learning the parameters of the network environments.

\begin{thm}
Given QoS exponents $\bm{\theta}  = \left\{ {{\theta_{n,k}},n = 1, \cdots ,N}, k = 1, \cdots ,K  \right\}$, the barring policy $\{ {d_{n,k}},n \in N,k \in K\} $ and the traffic arrival process, the transmit power $P_{\text{th}}^{n,k}$ of the $k$-th queue at the $n$-th device is obtained by solving
\begin{equation} \label{sita1}
\begin{aligned}
&{C_{n,k}}({\theta _{n,k}},{d_{n,k}},{\bm d_{ - n}},P_{\text{th}}^{n,k})\\
&= \frac{1}{{\theta _{n,k}^*{T_d}}}\log (\frac{{{p_n}}}{{1 - {\theta _{n,k}}{{\bar L}_{n,k}}}} + 1 - {p_n}).
\end{aligned}
\end{equation}
\end{thm}
Note that (\ref{sita1}) is obtained by substituting (\ref{EB}) and (\ref{EC}) into (\ref{Q1}), and can be solved using a bisection method, since ${C_{n,k}}({\theta _{n,k}},{d_{n,k}},{\bm d_{ - n}},P_{\text{th}}^{n,k})$ monotonically increases with $P_{\text{th}}^{n,k}$, the proof of which is provided as follows:

The first-order derivative of (\ref{rate}) with respect to ${P_{\text{th}}^{n,k}}$ is
\begin{equation}
\begin{aligned}
&\frac{{\partial {r_{n,k}}}}{{\partial P_{\text{th}}^{n,k}}} = \frac{{{\sigma _n}}}{{\ln 2(1 + P_{\text{th}}^{n,k}{\sigma _n})}}\\
&\left( {{\rm{1}} - \frac{{{Q^{ - 1}}({\varepsilon _{n,k}})}}{{\sqrt {{S_n}} }}\frac{1}{{\sqrt {{{\left( {1 + P_{\text{th}}^{n,k}{\sigma _n}} \right)}^4} - {{\left( {1 + P_{\text{th}}^{n,k}{\sigma _n}} \right)}^{\rm{2}}}} }}} \right)
\end{aligned}
\end{equation}
where ${\sigma _n} = \frac{{{G_n}{{\left| {{H_n}} \right|}^2}}}{{{N_0}}}$. When the number of OFDM symbols is sufficiently larger, we have $\sqrt {{S_n}}  \ge {Q^{ - 1}}({\varepsilon _{n,k}})$. Thus, $\frac{{\partial {r_{n,k}}}}{{\partial P_{\text{th}}^{n,k}}} \ge {\rm{0}}$. According to (\ref{EC}), we can readily prove that ${C_{n,k}}({\theta _{n,k}},{d_{n,k}},{\bm d_{ - n}},P_{\text{th}}^{n,k})$ increases with $r_{n,k}$. As a result, ${C_{n,k}}({\theta _{n,k}},{d_{n,k}},{\bm d_{ - n}},P_{\text{th}}^{n,k})$ increases with ${P_{\text{th}}^{n,k}}$.

\section{Barring Policy under QoS Guarantee}
\subsection{Distributed Game-Theoretic Barring Policy}
The activation probability of a queue determines the arrival rate and QoS of the traffic. The total performance in terms of effective capacity can be optimized by adaptively allocating the activation probability of all the queues over time based on the channel condition and the service delay requirement. Let ${C_{n,k}}({d_{n,k}},{\bm d_{ - n}}) \buildrel \Delta \over = {C_{n,k}}({\theta _{n,k}},{d_{n,k}},{\bm d_{ - n}},P_{\text{th}}^{n,k})$. As a result, the effective capacity maximization problem with constraints on barring policy is formulated as
\begin{equation}
\begin{aligned}
&(P1):  \mathop {\max }\limits_{\{ {d_{n,k}}\} } \sum\limits_{n \in \mathcal{N}} {\sum\limits_{k \in \mathcal{K}} {{C_{n,k}}} }({{d}_{n,k}},{\bm d_{ - n}}) \\
&s.t.~d_{min} \le {d_{n,k}} \le d_{max}
\end{aligned}
\end{equation}
where $\{ {d_{n,k}}\}$ is the activation probability of $k$-th queue at $n$-th device, and $d_{min}^{k}$ and $d_{max}^{k}$ are the maximum and minimum activation probability of the traffic class $k$, respectively. Due to the non-convexity of the (P1), it is difficult to obtain the optimal solution. Besides, accurate channel information is required for centralized computation, which is cost-ineffective in the mMTC system.

The successful access in its essence is network resource competition (e.g., preambles) among devices. As a result, it can be modeled as a non-cooperative game, which is effective to analyze the interactions among distributed decision makers and to improve the performance of decentralized networks \cite{new2}. Specifically, the priority-queuing-based access class barring policy across different devices is formulated as a non-cooperative game, denoted by ${\cal G} = \left[ {{\cal N},{{\left\{ {{{\bm d}_n}} \right\}}_{n \in {\cal N}}},{{\left\{ {{U_n}} \right\}}_{n \in {\cal N}}}} \right]$, where ${\cal N} = \{ 1, \ldots ,N\} $ is a set of players (i.e., devices), ${{U_n}}$ is a utility (payoff) function of player $n$, as given by
\begin{equation} \label{utility}
{U_n}({{\bm d}_n},{\bm d_{ - n}}) = \sum\limits_{k \in {\cal K}} {{C_{n,k}}\left( {{d_{n,k}},{\bm d_{ - n}}} \right) }.
\end{equation}
In (\ref{utility}), the effective capacity $C_{n,k}\left( {{d_{n,k}},{{\bm d}_{ - n}}} \right)$ can be regarded as the profit from channel access, given by
\begin{equation}
\begin{aligned}
{C_{n,k}}\left( {{d_{n,k}},{{\bm d}_{ - n}}} \right){\rm{ = }}& - \frac{1}{{{\theta _{n,k}}{T_s}}}\log \left\{ {1 - {d_{n,k}}{\Phi _{n,k}}({{\bm d}_{ - n}})} \right\}
\end{aligned}
\end{equation}
where
\begin{equation} \label{fa1}
\begin{aligned}
{\Phi _{n,k}}({{\bm d}_{ - n}}) &= \left( {1 - {E_{{H_n}}}\{ {e^{ - {\theta _{n,k}}{r_{n,k}}{S_{n,k}}}}\} } \right)(1 - {\varepsilon _n})\\
&\prod\limits_{j = 1}^{k - 1} {P_{\text{idle}}^{n,j}} (1 - P_{\text{idle}}^{n,k})\prod\limits_{l \ne n,l \in N} {\left( {1 - {D_l}\frac{1}{M}} \right)}.
\end{aligned}
\end{equation}


Given the game formulation, it is required to prove and derive the Nash equilibrium point. Since $C_{n,k}\left( {{d_{n,k}},{{\bm d}_{ - n}}} \right)$ increases with $d_{n,k}$ and there is no penalty for high activation probability, it is verified that the NE point is $\{ {d_{n,k}} = {d_{\max }}\} $. The NE can be far from the optimal solution of (P1).

In order to improve the efficiency of the game, we design the utility function of player $n$ based on pricing techniques \cite{7155564}, given by
\begin{equation}
{U_n}({\bm x_n},{\bm x_{ - n}}) = \sum\limits_{k \in {\cal K}} {\left[ {{C_{n,k}}\left( {{x_{n,k}},{{\bm x}_{ - n}}} \right) - {\lambda _{n,k}}{x_{n,k}}} \right]}
\end{equation}
where ${d_{n,k}} = 1 - {e^{ - {x_{n,k}}}}$, and $\lambda_{n,k}$ is regarded as a penalty price to prevent devices from always trying to access the channel and causing a heavy collision \cite{8334574}. To this end, we give Lemma 1 and Lemma 2, to prove the existence and uniqueness of NE of the new game $\mathcal{G}' = \left[ {\mathcal{N},{{\left\{ {{{\bm x}_n}} \right\}}_{n \in \mathcal{N}}},{{\left\{ {{U_n}} \right\}}_{n \in \mathcal{N}}}} \right]$.
\begin{lem}
${U_n}({{\bm x}_n},{{\bm x}_{ - n}})$ is a concave function of ${{\bm x}_n}$.
\end{lem}
\begin{proof}
The first-order derivative of ${U_n}$ with respect to $x_{n,k}$ is
\begin{equation} \label{derivative}
\frac{{\partial {U_n}}}{{\partial {x_{n,k}}}} = \frac{{{e^{ - {x_{n,k}}}}{\Phi _{n,k}}({{\bm x}_{ - n}})}}{{{\theta _{n,k}}{T_s}\left( {1 - {\Phi _{n,k}}({{\bm x}_{ - n}}) + {e^{ - {x_{n,k}}}}{\Phi _{n,k}}({{\bm x}_{ - n}})} \right)}}-\lambda _{n,k}
\end{equation}
and the second-order derivative of ${U_n}$ with respect to $x_{n,k}$ is
\begin{equation}
\frac{{{\partial ^2}{U_n}}}{{{\partial ^2}{x_{n,k}}}}{\rm{ = }}\frac{{{e^{{x_{n,k}}}}\left( {{\Phi _{n,k}}({{\bm x}_{ - n}}) - 1} \right){\Phi _{n,k}}({{\bm x}_{ - n}})}}{{{\theta _{n,k}}{T_s}{{\left( {{e^{{x_{n,k}}}} + {\Phi _{n,k}}({{\bm x}_{ - n}}) - {e^{{x_{n,k}}}}{\Phi _{n,k}}({{\bm x}_{ - n}})} \right)}^2}}}.
\end{equation}
According to (\ref{fa1}), we have $0 < {\Phi _{\rm{1}}}({{\bm x}_{ - n}}) < 1$, and $\frac{{{\partial ^2}{U_n}}}{{{\partial ^2}{x_{n,k}}}} < 0$. The second-order partial derivative of $U_n$, $\frac{{{\partial ^{\rm{2}}}{U_n}}}{{\partial {x_{n,k}}\partial {x_{n,k'}}}}$, is zero. As a result, the Hessian matrix of $U_n$ is negative. Thus, this concludes the proof.
\end{proof}

\begin{lem}
  ${U_n}({{\bm x}_n},{{\bm x}_{ - n}})$ is a sub-modular function.
\end{lem}
\begin{proof}
  The second-order derivative of $U_n$ with respect to $x_{n,k}$ and $x_{j,k'}$, $j \ne n$, is given by
  \begin{equation}
  \begin{aligned}
  \frac{{\partial {U_n}}}{{\partial {x_{n,k}}\partial {x_{j,k'}}}} = \frac{\partial }{{\partial {x_{j,k'}}}}\left( {\frac{{{e^{ - {x_{n,k}}}}}}{{{\theta _{n,k}}{T_s}\left( {\frac{1}{{{\Phi _{n,k}}({{\bm x}_{ - n}})}} - 1 + {e^{ - {x_{n,k}}}}} \right)}}} \right)\\
   = \frac{{{e^{ - {x_{n,k}}}}}}{{{\theta _{n,k}}{T_s}{{\left( {\frac{1}{{{\Phi _{n,k}}({{\bm x}_{ - n}})}} - 1 + {e^{ - {x_{n,k}}}}} \right)}^2}}}\frac{1}{{{\Phi _{n,k}}{{({{\bm x}_{ - n}})}^2}}}\frac{{\partial {\Phi _{n,k}}({{\bm x}_{ - n}})}}{{\partial {x_{j,k'}}}}.
  \end{aligned}
  \end{equation}
  According to (\ref{fa1}), we have $\frac{{\partial {\Phi _{n,k}}({{\bm x}_{ - n}})}}{{\partial {x_{j,k'}}}} < 0$. As a result, $\frac{{{\partial ^2}{U_n}}}{{\partial {x_{n,k}}\partial {x_{j,k'}}}} < 0$. Since ${U_n}({{\bm x}_n},{\bm x_{ - n}})$ is twice differentiable, and $\frac{{{\partial ^2}{U_n}}}{{\partial {x_{n,k}}\partial {x_{j,k'}}}} < 0$, the utility function ${U_n}({{\bm x}_n},{{\bm x}_{ - n}})$ is a sub-modular function \cite{7470936}.
\end{proof}
Informally, the sub-modularity of the game implies that the best-response of player $n$ is a non-increasing function of another players' decision \cite{7470936}. The strict mathematic proof is given in the following.

\begin{thm}[NE Existence]
 The game $\mathcal{G}' = \left[ {\mathcal{N},{{\left\{ {{{\bm x}_n}} \right\}}_{n \in \mathcal{N}}},{{\left\{ {{U_n}} \right\}}_{n \in \mathcal{N}}}} \right]$ has at least one pure-strategy Nash equilibrium $\{ x_{n,k}^*,\forall n \in \mathcal{N},\forall k \in \mathcal{K}\}$.
\end{thm}
\begin{proof}
For every player, $d_{n,k}$ is within the interval $[d_{\min}, d_{\max}]$. Thus, we have $- \ln (1 - {d_{\min }}) \le {x_{n,k}} \le  - \ln (1 - {d_{\max }})$. Let $x_{\min} = - \ln (1 - {d_{\min }})$ and ${x_{\max }} =  - \ln (1 - {d_{\max }})$. Obviously, set ${x_{n,k}} \in [{x_{\min }},{x_{\max }}]$ is compact and convex. The utility function ${U_n}({{\bm x}_n},{x_{ - n}})$ is a continuous function of probability vector ${{\bm x}_n}$. In addition, ${U_n}({{\bm x}_n},{\bm x_{ - n}})$ is a concave function of ${{\bm x}_n}$ (See Lemma 1). According to theorem 3.2 in \cite{han2012game}, we prove the existence of Nash equilibrium.
\end{proof}
\begin{thm}
The best-response function of player $n$ is given by
\begin{equation} \label{bestresponse}
\begin{aligned}
x_{n,k} &= b_{n,k}({\bm x_{-n}}) \\ &=\left[ {\ln \!\left( {\frac{1}{{{\lambda _{n,k}}{\theta _{n,k}}{T_s}}}\! -\! 1\!} \right) \!\!- \!\! \ln \!\left( {\frac{1}{{{\Phi _{n,k}}({{\bm x}_{ - n}})}} \!-\! 1\!} \right)} \right]_{{x_{\min }}}^{{x_{\max }}}
\end{aligned}
\end{equation}
and is a decreasing function of ${\bm x}_{-n}$, where $[x]_{{x_{\min }}}^{{x_{\max }}} = \min \left[ {\max \left( {x,{x_{\min }}} \right),{x_{\max }}} \right]$.
\end{thm}
\begin{proof}
Intuitively, the best-response function of player $n$ is the point at which the derivative of ${U_n}({{\bm x}_n},{{\bm x}_{ - n}})$ with respect to $x_{n,k}$ equals zero. According to (\ref{derivative}), it can be calculated by
\begin{equation} \label{zero}
\frac{{{e^{ - {x_{n,k}}}}{\Phi _{n,k}}({{\bm x}_{ - n}})}}{{{\theta _{n,k}}{T_s}\left( {1 - {\Phi _{n,k}}({{\bm x}_{ - n}}) + {e^{ - {x_{n,k}}}}{\Phi _{n,k}}({{\bm x}_{ - n}})} \right)}} = {\lambda _{n,k}}.
\end{equation}
After some algebraic manipulation, (\ref{bestresponse}) can be obtained, which decreases with ${\bm x}_{ - n}$. This is because
\[\frac{{\partial b_{n,k}({\bm x_{-n}})}}{{\partial {x_{j,k'}}}} = \frac{1}{{{\Phi _{n,k}}({{\bm x}_{ - n}}) - {\Phi _{n,k}}{{({{\bm x}_{ - n}})}^2}}}\frac{{\partial {\Phi _{n,k}}({{\bm x}_{ - n}})}}{{\partial {x_{j,k'}}}}.\]
where $0 < {\Phi _{n,k}}({{\bm x}_{ - n}}) < 1$ and $\frac{{\partial {\Phi _{n,k}}({{\bm x}_{ - n}})}}{{\partial {x_{j,k'}}}} < 0$, and we have $\frac{{\partial b_{n,k}({\bm x_{-n}})}}{{\partial {x_{j,k'}}}} < 0$. The proof is accomplished.
\end{proof}

Note that Theorem 3 provides the best-response function of device $n$, which needs its information (including the penalty price $\lambda _{n,k}$, the QoS requirement $\theta _{n,k}$, and the channel condition $H_n$), and also needs exchange information (more specifically, the activation probability ${\vec x_{-n}}$) to calculate $\Phi _{n,k}({{\vec x}_{ - n}})$. Assume that $x_{n,k}$ can be quantified by $s$ bits. As a result, the total overhead at a iteration could has $KNs$. To further reduce the overhead, we propose the following fully distributed algorithm.

According to Lemma 2 and Theorem 3, a greedy updating algorithm is developed to iteratively search the Nash equilibrium point. At iteration $t$, each player chooses the best strategy in respond to the opponent strategy chosen in iteration $t-1$. Note that the iteration times do not correspond to the time slots of our proposed protocol. The duration of iterations can be very smaller than the superframe duration, as long as the exchange information is delivered. For the game of the minimization of sub-modular cost function, the greedy updating algorithm is proved to converge monotonically to an equilibrium (the equilibrium depends on the initial state of the algorithm) \cite{1198610,topkis1979equilibrium}. However, the conclusion cannot fit to the maximization of sub-modular utility function. In the following, we prove the convergence of the algorithm in the case of maximization of sub-modular utility function, and the uniqueness of the equilibrium.

\begin{thm} [Convergence and Uniqueness]
  Consider a non-cooperative game,  where the payoff functions are sub-modular, players maximize the objective function, and the strategy spaces are continuous and compact. Then, the following hold:
  \begin{enumerate}
    \item If player $n$ initially uses the policy satisfying ${x_{n,k}^0} \le b_{n,k}({\bm x_{ - n}^*})$ and ${x_{n,k}^0} < {x_{n,k}^2}$ (${x_{n,k}^0} \ge b_{n,k}({\bm x_{ - n}^*})$ and ${x_{n,k}^0} > {x_{n,k}^2}$), $\forall n \in \mathcal{N},\forall k \in \mathcal{K}$, the sequence ${\{ x_{n,k}^t\} _{t = 0,1,2, \ldots }}$ generated by the greedy updating algorithm converges to an Nash equilibrium $\{ x_{n,k}^*,\forall n \in \mathcal{N},\forall k \in \mathcal{K}\}$.

    \item The Nash equilibrium $\{ x_{n,k}^*,\forall n \in \mathcal{N},\forall k \in \mathcal{K}\}$ is unique.

  \end{enumerate}
\end{thm}
\begin{proof}
  We first prove item 1). According to the existence theorem, there is a Nash equilibrium, denoted by $\{ x_{n,k}^*,\forall n \in \mathcal{N},\forall k \in \mathcal{K}\}$, satisfying $x_{n,k}^* = b_{n,k}(\bm x_{ - n}^*)$. The initial point is assume to satisfy $x_{n,k}^0 \le b_{n,k}(\bm x_{ - n}^*),\forall n \in {\cal N},\forall k \in {\cal K}$. Due to decreasing property of the best-response function,  we have
  \begin{equation} \label{proof1}
  \begin{array}{l}
    x_{n,k}^0 \le b_{n,k}(\bm x_{ - n}^*) = \bm x_{ - n}^*,\forall n \in {\cal N},\forall k \in {\cal K}\\
    x_{n,k}^1 = b_{n,k}(\bm x_{ - n}^0) \ge b_{n,k}(\bm x_{ - n}^*) = \bm x_{ - n}^*,\forall n \in {\cal N},\forall k \in {\cal K}\\
    x_{n,k}^2 = b_{n,k}(\bm x_{ - n}^1) \le b_{n,k}(\bm x_{ - n}^*) = \bm x_{ - n}^*,\forall n \in {\cal N},\forall k \in {\cal K}\\
    x_{n,k}^3 = b_{n,k}(\bm x_{ - n}^2) \ge b_{n,k}(\bm x_{ - n}^*) = \bm x_{ - n}^*,\forall n \in {\cal N},\forall k \in {\cal K}\\
     \cdots \\
    x_{n,k}^{2t} = b_{n,k}(\bm x_{ - n}^{2t - 1}) \le b_{n,k}(\bm x_{ - n}^*) = \bm x_{ - n}^*,\forall n \in {\cal N},\forall k \in {\cal K}\\
    x_{n,k}^{2t + 1} = b_{n,k}(\bm x_{ - n}^{2t}) \ge b_{n,k}(\bm x_{ - n}^*) = \bm x_{ - n}^*,\forall n \in {\cal N},\forall k \in {\cal K}.
\end{array}
\end{equation}
   As a result, we can observe that the sequence ${\{ x_{n,k}^t - x_{n,k}^*\} _{t = 0,1,2, \ldots }},\forall n \in \mathcal{N},\forall k \in \mathcal{K},$ is an alternating sequence whose terms alternate in sign. Stated otherwise, the original sequence ${\{ x_{n,k}^t\} _{t = 0,1,2, \ldots }}$ oscillates at the Nash equilibrium $\{ x_{n,k}^*\}$. In addition, we can observe that if $x_{n,k}^2 - x_{n,k}^0 > 0$, the following inequations are satisfied:
  \begin{equation} \label{proof2}
    \begin{array}{l}
    x_{n,k}^3 - x_{n,k}^1 = b_{n,k}(\bm x_{ - n}^2) - b_{n,k}(\bm x_{ - n}^0) < 0,\forall n \in {\cal N},\forall k \in {\cal K}\\
    x_{n,k}^4 - x_{n,k}^2 = b_{n,k}(\bm x_{ - n}^3) - b_{n,k}(\bm x_{ - n}^1) > 0,\forall n \in {\cal N},\forall k \in {\cal K}\\
     \cdots \\
    x_{n,k}^{2t - 1}\!\! -\!\! \bm x_{ - n}^{2t - 3} \!\!=\!\! b_{n,k}(\bm x_{ - n}^{2t - 2}) - b_{n,k}(\bm x_{ - n}^{2t - 4}) < 0,\forall n \in {\cal N},\forall k \in {\cal K}\\
    x_{n,k}^{2t}\!\! -\!\! \bm x_{ - n}^{2t - 2} \!\!= \!\!b_{n,k}(\bm x_{ - n}^{2t - 1}) - b_{n,k}(\bm x_{ - n}^{2t - 3}) > 0,\forall n \in {\cal N},\forall k \in {\cal K}.
\end{array}
\end{equation}
According to (\ref{proof1}) and (\ref{proof2}), the sub-sequence ${\{ x_{n,k}^{2t}\} _{t = 0,1,2, \ldots }},\forall n \in \mathcal{N},\forall k \in \mathcal{K},$ increases with $t$, and the upper bound is $\{ x_{n,k}^*\}$. Similarly, the sub-sequence ${\{ x_{n,k}^{2t-1}\} _{t = 1,2, \ldots }},\forall n \in \mathcal{N},\forall k \in \mathcal{K},$ decreases with $t$, and the lower bound is $\{ x_{n,k}^*\}$. As a result, there is a positive integer $t$ such that for $\forall n \in \mathcal{N},\forall k \in \mathcal{K}$, we have $\left| {x_{n,k}^t - \bm x_{ - n}^{t - 1}} \right| \le \zeta $, where $\zeta$ is a small positive real number and the vertical bars denote the absolute value. Stated otherwise, the original sequence ${\{ x_{n,k}^t\} _{t = 0,1,2, \ldots }}$ is a Cauchy sequence, which converges to the Nash equilibrium.

As discussed precedingly, the convergence of the algorithm needs to be guaranteed, satisfying that $x_{n,k}^2 - x_{n,k}^0 > 0$. Thus, we redesign the iteration, given by
\begin{equation}
x_{n,k}^2 = \left\{ {\begin{array}{*{20}{c}}
{b_{n,k}(\bm x_{ - n}^1),{\rm{when }}~~b_{n,k}(\bm x_{ - n}^1) > x_{n,k}^0}\\
{x_{n,k}^0 + \delta ,{\rm{when }}~~b_{n,k}(\bm x_{ - n}^1) \le x_{n,k}^0}
\end{array}} \right.
\end{equation}
where $\delta$ is a small positive real number. For the case of ${x_{n,k}^0} \ge b_{n,k}({\bm x_{ - n}^*})$ and ${x_{n,k}^0} > {x_{n,k}^2}$, similar proof can be obtained.

Next, we proceed to prove item 2). Suppose $\{ x_{n,k}^*\} $ and $\{ x_{n,k}^{**}\} $ are distinct Nash equilibriums, satisfying $x_{n,k}^{*} = b_{n,k}\left( {\bm x_{ - n}^{*}} \right)$ and $x_{n,k}^{**} = b_{n,k}\left( {\bm x_{ - n}^{**}} \right)$, respectively. Without loss of generality, we assume there exist $n$ and $k$ such that $x_{n,k}^{**} < x_{n,k}^*$. Hence, there exists $\alpha  > 1$ such that for some $n$ and $k$, $\alpha x_{n,k}^{**} > x_{n,k}^*$ and that for other $n$ and $k$, $\alpha x_{n,k}^{**} = x_{n,k}^*$. Thus, with $n$ and $k$ satisfying $\alpha x_{n,k}^{**} = x_{n,k}^*$, we have
\begin{equation}
\alpha b_{n,k}\left( {x_{n,k}^{**}} \right) = b_{n,k}\left( {x_{n,k}^*} \right).
\end{equation}
Since there exists at least one $n'$ and $k'$ such that $\alpha x_{n',k'}^{**} > x_{n',k'}^*$, we have $\alpha \bm x_{ - n}^{**} > \bm x_{ - n}^*$. According to the decreasing property of the best-response function, we have
\begin{equation}
b_{n,k}\left( {\alpha \bm x_{ - n}^{**}} \right) < b_{n,k}\left( {\bm x_{ - n}^*} \right).
\end{equation}
Due to contradiction, we infer the uniqueness of Nash equilibrium.
\end{proof}

As the algorithm requires all the devices' information, we propose the following fully distributed algorithm. Substituting (\ref{12}) into (\ref{fa1}), we have
\begin{equation} \label{dis}
{\Phi _{n,k}}({\bm x_{ - n}}) = \left( {1 - {\mathbb{E}_{{H_n}}}\{ {e^{ - {\theta _{n,k}}{r_{n,k}}{S_{n,k}}}}\} } \right)\frac{{\tilde F_s^{n,k}}}{{1 - {e^{ - {x_{n,k}}}}}}
\end{equation}
where ${\tilde F_s^{n,k}} = (1 - {\varepsilon _n})F_s^{n,k}$ is an estimated successful access probability, and can be obtained by observing the environment, i.e., counting the acknowledgement (ACK) sent by the BS \cite{new7}. In this way, we do not need to estimate the empty buffer probability and exchange the information with other devices. The channel state information (CSI) can be  estimated at the BS, and usually quantized and fed back to the device. Substituting (\ref{dis}) into (\ref{bestresponse}), we obtain a fully distributed best-response function, requiring local information, such as current channel state and the estimated successful access probability. Thus, each device iterates by itself using the fully distributed best-response function, instead of exchanging information among devices. The proposed algorithm is explicitly described in Algorithm 1.
\begin{algorithm}[t]
  \caption{A distributed best-response-based algorithm for the $k$-th queue of the $n$-th device}
  \label{Algorithm1}
  \begin{algorithmic}[1]
  \STATE Initialize $x_{n,k}[0]$, $t=0$ and the error $\xi$.
  \REPEAT
  \STATE Estimate current channel state and the successful access probability. Update ${\Phi _{n,k}}({\bm x_{ - n}})$ using (\ref{dis}).
  \STATE $t \leftarrow t + 1$.
  \STATE Update $x_{n,k}[t]$ using (\ref{bestresponse}).
  \UNTIL{$\left| {{x_{n,k}}[t] - {x_{n,k}}[t - 1]} \right| \le \xi$}
  \end{algorithmic}
\end{algorithm}
\subsection{Local Optimal Price}
The price selection $\lambda_{n,k}$ is of importance to maximize the effective capacity. A smaller $\lambda_{n,k}$ implies frequent access attempts, making the signalling more conflicted. However, a larger $\lambda_{n,k}$ prevents more devices from accessing the BS. Hence, an appropriate price is necessary to the game.

The local optimal price can be obtained by comparing the Lagrangian of the centralized optimization problem (P1) and that of the maximization of each player's utility. The problem can be reformulated as
\begin{equation}
\begin{aligned}
&(P2):\mathop {\min}\limits_{\{ {x_{n,k}}\} } -\sum\limits_{n \in {\cal N}} {\sum\limits_{k \in {\cal K}} {{C_{n,k}}\left( {{x_{n,k}},{{\bm x}_{ - n}}} \right)} } \\
&s.t.\;{x_{min}} \le {x_{n,k}} \le {x_{max}}.
\end{aligned}
\end{equation}
The Lagrangian function of (P2) is given by
\begin{equation}
\begin{aligned}
{L_1} = &\sum\limits_{n \in {\cal N}} {\sum\limits_{k \in {\cal K}} {\left\{ { - {C_{n,k}}\left( {{x_{n,k}},{{\bm x}_{ - n}}} \right) + ({\sigma _{n,k}} - {\nu _{n,k}}){x_{n,k}}} \right.} } \\
&~~~~~~~~~~~\left. { + {\nu _{n,k}}{x_{min}} - {\sigma _{n,k}}{x_{max}}} \right\}
\end{aligned}
\end{equation}
where ${\sigma _{n,k}}$ and ${\nu _{n,k}}$ are the Lagrange multipliers for the constraints. Let ${\bm x^{opt}}$ be the optimal device activation probability. The first order KKT condition is given by
\begin{equation} \label{KKT1}
\begin{aligned}
{\left. {\frac{{\partial {L_1}}}{{\partial {x_{n,k}}}}} \right|_{\bm x = {\bm x^{opt}}}} =&  - {\left. {\frac{{\partial {C_{n,k}}\left( {{x_{n,k}},{{\bm x}_{ - n}}} \right)}}{{\partial {x_{n,k}}}}} \right|_{\bm x = {\bm x^{opt}}}} + \sigma _{n,k} - \nu _{n,k}\\
& - \!\!\!\!\!\!\!\!{\left. {\sum\limits_{m \ne n,m \in {\cal N}} {\sum\limits_{b \in {\cal K}} {\frac{{\partial {C_{m,b}}\left( {{x_{m,b}},{{\bm x}_{ - m}}} \right)}}{{\partial {x_{n,k}}}}} } } \right|_{\bm x = {\bm x^{opt}}}} =0.
\end{aligned}
\end{equation}
The maximization of each player's utility can be formulated as
\begin{equation}
\begin{aligned}
(P3):\mathop {\max }\limits_{\{ {x_{n,k}}\} } {U_n}\left( {{x_{n,k}},{{\bm x}_{ - n}}} \right)\\
s.t.\;{x_{min}} \le {x_{n,k}} \le {x_{max}}.
\end{aligned}
\end{equation}
The Lagrangian function of (P3) is given by
\begin{equation}
\begin{aligned}
{L_2} = &\sum\limits_{k \in {\cal K}} {\left\{ { - \left( {{C_{n,k}} - {\lambda _{n,k}}{x_{n,k}}} \right) + ({{\sigma '}_{n,k}} - {{\nu '}_{n,k}}){x_{n,k}}} \right.} \\
&~~~~~~\left. { + {{\nu '}_{n,k}}{x_{min}} - {{\sigma '}_{n,k}}{x_{max}}} \right\}
\end{aligned}
\end{equation}
where ${\sigma'_{n,k}}$ and ${\nu'_{n,k}}$ are the Lagrange multipliers for the constraints. Let ${\bm x^{*}}$ be the NE point. The first order KKT condition for each player's utility at the NE point is
\begin{equation} \label{KKT2}
\begin{aligned}
{\left. {\frac{{\partial {L_2}}}{{\partial {x_{n,k}}}}} \right|_{\bm x = {\bm x^{*}}}} \!\!=\!\! & - {\left. {\frac{{\partial {C_{n,k}}{\left( {{x_{n,k}},{{\bm x}_{ - n}}} \right)}}}{{\partial {x_{n,k}}}}} \right|_{\bm x = {\bm x^{*}}}} + {\lambda _{n,k}} \\
&+ {{\sigma '}_{n,k}} - {{\nu '}_{n,k}}{\rm{ = 0}}.
\end{aligned}
\end{equation}

By comparing (\ref{KKT1}) with (\ref{KKT2}), we find that they are equivalent if
\begin{equation} \label{optimalprice}
{\lambda _{n,k}} =  - {\left. {\sum\limits_{m \ne n,m \in {\cal N}} {\sum\limits_{b \in {\cal K}} {\frac{{\partial {C_{m,b}}\left( {{x_{m,b}},{{\bm x}_{ - m}}} \right)}}{{\partial {x_{n,k}}}}} } } \right|_{\bm x = {\bm x^{opt}}}}.
\end{equation}
Denote
\begin{equation} \label{optimalprice2}
\begin{aligned}
&f(\bm x) =  - \sum\limits_{m \ne n,m \in {\cal N}} {\sum\limits_{b \in {\cal K}} {\frac{{\partial {C_{m,b}}\left( {{x_{m,b}},{{\bm x}_{ - m}}} \right)}}{{\partial {x_{n,k}}}}} }\\
&= \frac{1}{{{T_s}}}\frac{1}{{1 - {D_n}\frac{1}{M}}}\frac{1}{M}\prod\limits_{j = 1}^{k - 1} {P_{\text{idle}}^{n,j}} (1 - P_{\text{idle}}^{n,k}){e^{ - {x_{n,k}}}}\\
&\sum\limits_{m \ne n,m \in {\cal N}} {\sum\limits_{b \in {\cal K}} {\frac{1}{{{\theta _{m,b}}}}\frac{{(1 - {e^{ - {x_{m,b}}}}){\Phi _{m,b}}({{\bm x}_{ - m}})}}{{1 - (1 - {e^{ - {x_{m,b}}}}){\Phi _{m,b}}({{\bm x}_{ - m}})}}} }.
\end{aligned}
\end{equation}
Note that (36) needs some exchange information (more specifically, the activation probability ${\bm{x}_{-n}}$ and the channel condition $H_n$). We assume BS can assist to update $\lambda _{n,k}$. Thus, the channel condition only exchange one time with the BS. And at each iteration, the activation probability ${\bm{x}_{-n}}$ should be exchange with the BS, and the BS seed $\lambda _{n,k}$ to devices to improve the performance.

We can update the price as follows:
\begin{equation} \label{neweq}
{\lambda _{n,k}}[t + 1] = (1 - {\rho _t}){\lambda _{n,k}}[t] + {\rho _t}{f_{n,k}}({\bf{x}}),
\end{equation}
where the step-size $\rho$ is a decreasing sequence.

Through updating the price using (\ref{neweq}) and the best-response function using (\ref{bestresponse}), a price-update algorithm is developed to obtain the local optimality, as described in Algorithm 2. The convergence proof of Algorithm 2 is challenging considering the dynamic price in the game. We can provide a undemanding proof of convergence for our algorithm. When the number of iterations is large, the price ${\lambda _{n,k}}[t]$ tends to steady. As a result, according to Theorem 4, given a fixed price, the algorithm can fast converge to the NE point. Note that, due to the non-convexity of (P2),  Algorithm 2 only solves for local optimal solutions. In other words, there are some saddle points or KKT points of (P2), and the final solution of Algorithm 2 depends on the selection of initial points.

\begin{algorithm}[t]
  \caption{A price-update algorithm}
  \label{Algorithm1}
  \begin{algorithmic}[1]
  \STATE Initialize $x_{n,k}[0]$, $t=0$ and error bound $\xi$.
  \REPEAT
  \STATE $t \leftarrow t + 1$.
  \STATE Update the price using (\ref{neweq}).
  \FOR {$n \in \mathcal{N},k \in \mathcal{K}$}
  \STATE Update $x_{n,k}[t]$ using (\ref{bestresponse}).
  \ENDFOR
  \UNTIL{$\left| {\bm x[t] - \bm x[t - 1]} \right| \le \xi $}
  \end{algorithmic}
\end{algorithm}

\section{Simulation Results}
In this section, numerical results are provided. First, we evaluate the performance of the mMTC system using the proposed effective capacity. Then, the convergence and effectiveness of the distributed barring policy for maximizing each device's utility is demonstrated. Finally, we provide the results using the proposed price-update algorithm for total effective capacity maximization.

In our simulations, the BS is located in the center of an area of $500\times500$ m$^2$ and the devices are randomly and uniformly distributed with in this area. The number of traffic types is $K=2$, and the average packet length is $L=500$ bits. The QoS requirement of traffic type 1 is ${\theta _{n,1}} = {10^{ - 3}},n \in \mathcal{N}$, and that of traffic type 2 is ${\theta _{n,2}} = {10^{ - 5}},n \in \mathcal{N}$. Set $d_{min}^{k} = 0.1$ and $d_{max}^{k} = 0.9, k=1,2$. The activating probability of traffic type 1 is ${d_{n,1}} =0.9,n \in \mathcal{N}$, and that of traffic type 2 is ${d_{n,2}} = 0.5,n \in \mathcal{N}$. The duration of a transmission frame is $T_u = 4$ ms, and that of data transmission is $T_f=3$ ms. The slot $T_d$ is $0.5$ ms.The transmit power of each device is 10 dBm. The noise spectral density is $-174$ dBm/Hz. The large-scale path loss is $PL_n = 60+37.6lg(X_n)$, where $X_n$ is the distance from the $n$-th device to the BS. The packet error probability $\varepsilon_{n,k} = 10^{-5}$. Each point of results is the average of 200 independent runs.

\subsection{Performance analysis of the proposed effective capacity}
In Fig. \ref{fig3}, we evaluate the joint impact of the preambles resources and data transmission resources on the effective capacity of mMTC. The number of devices is $N=100$. The duration of a symbol is $a = 66.7 $~$\mu s$ and the bandwidth of a symbol is set to $15$ KHz. We can evaluate the effective capacity by varying preamble number $M$ and transmission bandwidth $B_n$.  We observe that the effective capacity grows with the preamble resources and the bandwidth in both cases of traffic types. Hence, we can joint allocate the two resources to meet the QoS requirements of devices.

\begin{figure}[!t]
\centering
\includegraphics[width=3.2in]{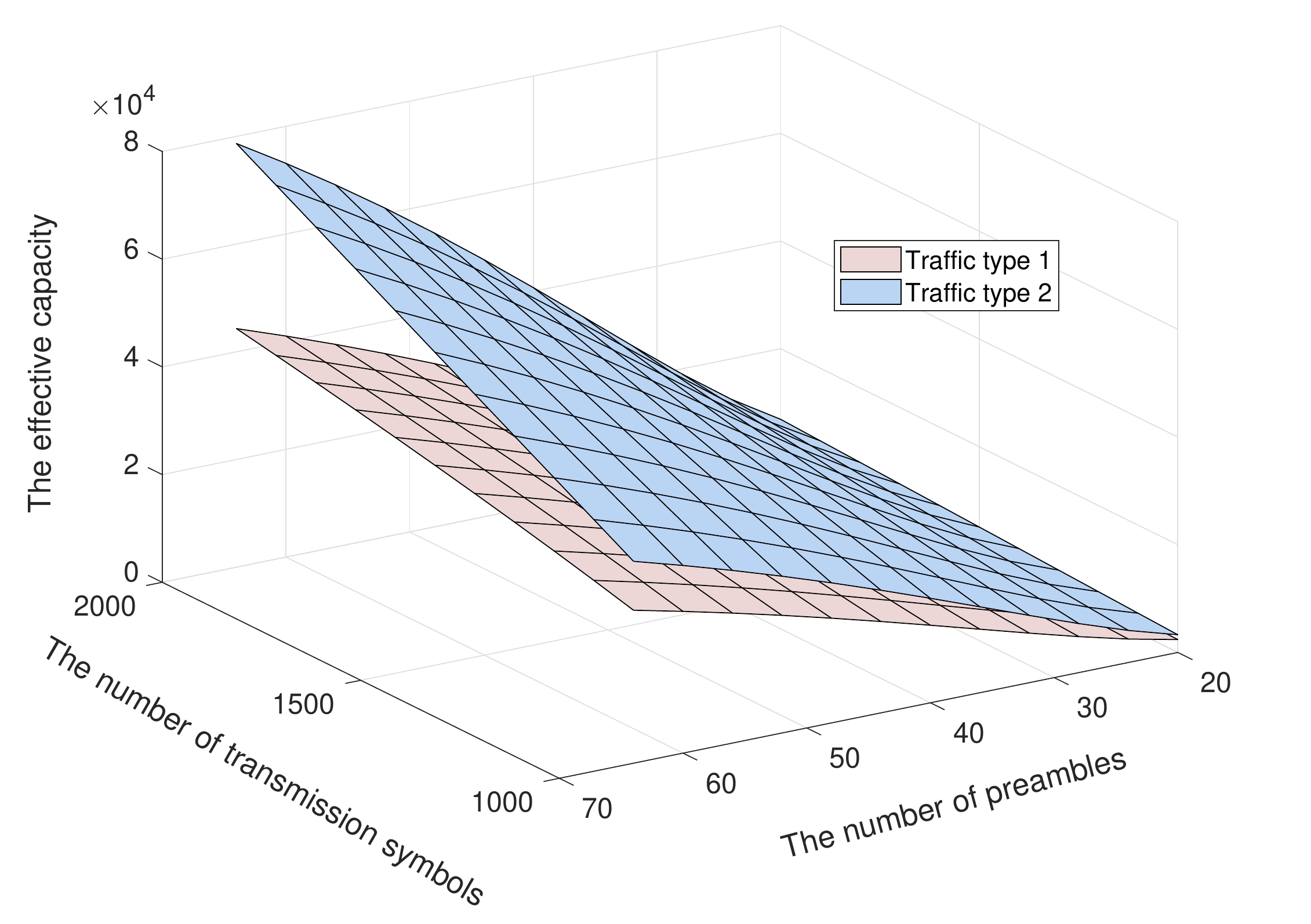}
\vspace{-0.4cm}
\caption{The average effective capacity of one device versus the preambles resources and the data transmission resources.}
\label{fig3}
\vspace{-0.2cm}
\end{figure}

Fig. \ref{fig4} shows the performance of mMTC in terms of different QoS requirements. In Fig. \ref{fig4}(a), we evaluate the average effective capacity with an increasing QoS exponent of traffic type 1, given $\theta_{n,2}=10^{-6}$. It is shown that the effective capacity of traffic type 1 decreases with QoS exponent of traffic type 1, ${\theta _{n,1}}$, and is sensitive to ${\theta _{n,1}}$ within the region ${10^{ - 4}} \le {\theta _{n,1}} \le {10^{ - 2}}$. On the contrary, the effective capacity of traffic type 2 decreases with ${\theta _{n,1}}$. This is because the larger ${\theta _{n,1}}$, the larger the queue idle probability of traffic type 1. Thus, the traffic type 2 has a higher access probability. In Fig. \ref{fig4}(b), we evaluate the average effective capacity with an increasing QoS exponent of traffic type 2, given $d_{n,1}=10^{-3}$. The observation is similar to that from Fig. \ref{fig4}(a).

\begin{figure}[!t]
\centering
\includegraphics[width=3.5in]{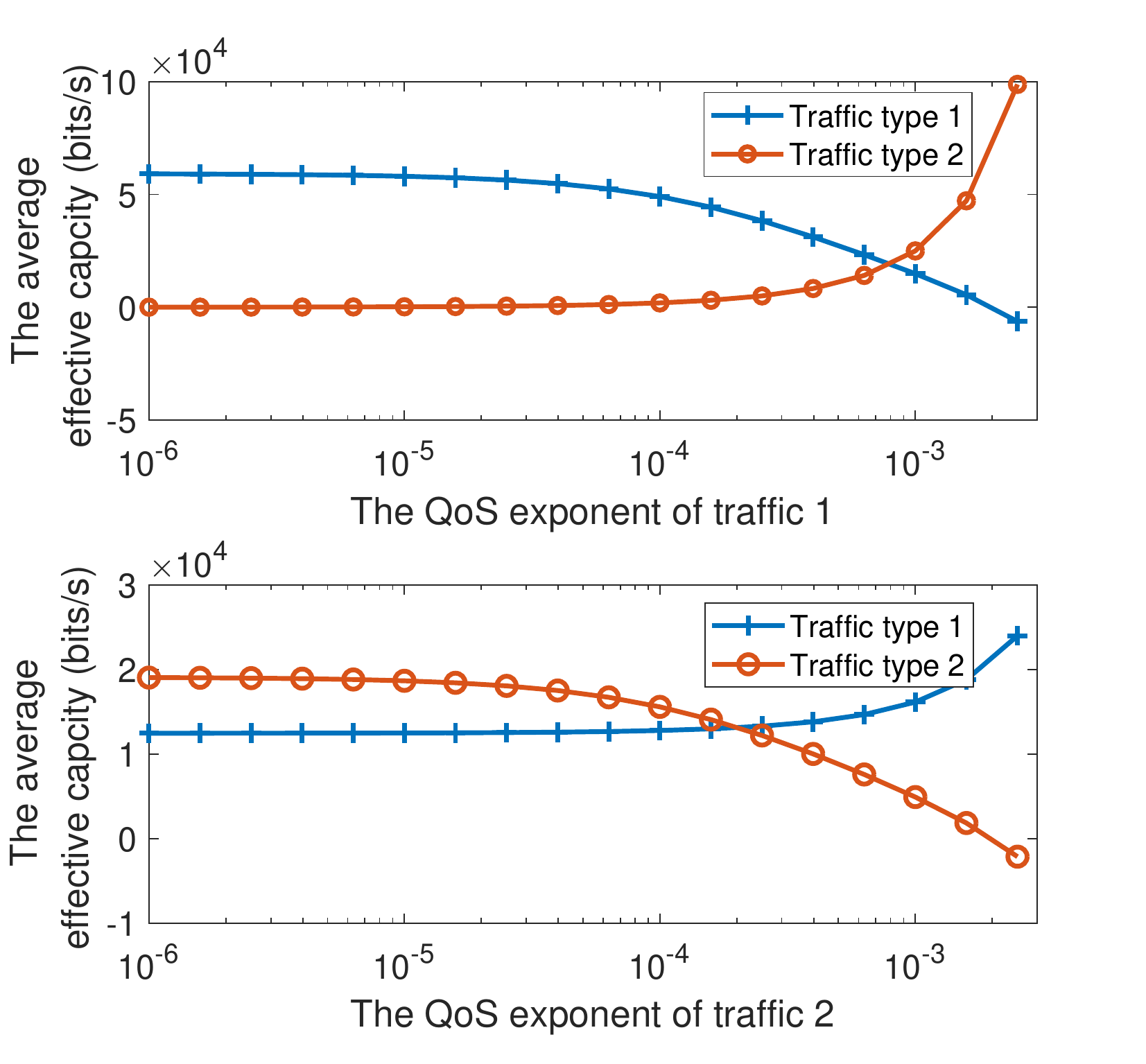}
\vspace{-0.4cm}
\caption{The average effective capacity of one device versus different QoS requirements: (a) Given QoS requirement of traffic type 2, $d_{n,2}=10^{-6}$; (b) Given QoS requirement of traffic type 1, $d_{n,2}=10^{-3}$.}
\label{fig4}
\vspace{-0.2cm}
\end{figure}

\subsection{Convergence and performance of the distributed barring policy}
In Fig. \ref{fig5}, we investigate the convergence of Algorithm 1, where ${\lambda _{n,k}} = 1000$, ${S_{n,k}} = 1000$, and $M=50$. We assume that the successful access probability of each device can be accurately estimated. In Fig \ref{fig5}, we can observe that Algorithm 1 converges fast to the same NE point regardless of the initial point, which can be verified by the property of the game, as discussed in Section IV.A.
\begin{figure}[!t]
\centering
\includegraphics[width=3.2in]{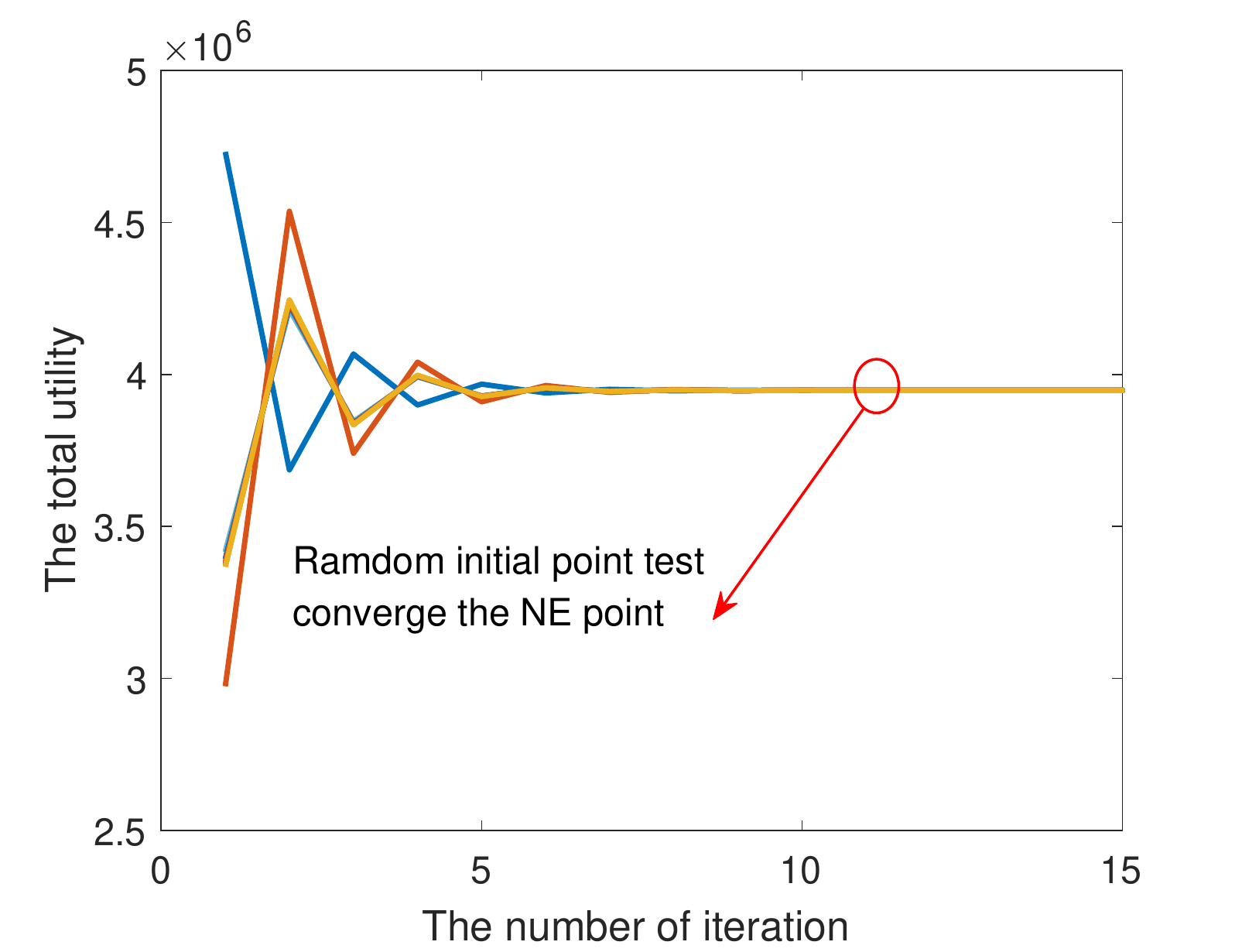}
\vspace{-0.4cm}
\caption{Convergence of Algorithm 1.}
\label{fig5}
\vspace{-0.2cm}
\end{figure}

Fig. \ref{fig6} shows the total effective capacity of Algorithm 1 with different penalty prices, compared with the fixed activating probability. The simulation parameters are the same as those for Fig. \ref{fig5}. It is shown that Algorithm 1 is superior to other strategies. This is because our proposed algorithm can search the NE point. We also observe that the penalty price has an important influence on the total effective capacity, as mentioned earlier. There is an optimal price to maximize the total effective capacity.
\begin{figure}[!t]
\centering
\includegraphics[width=3.2in]{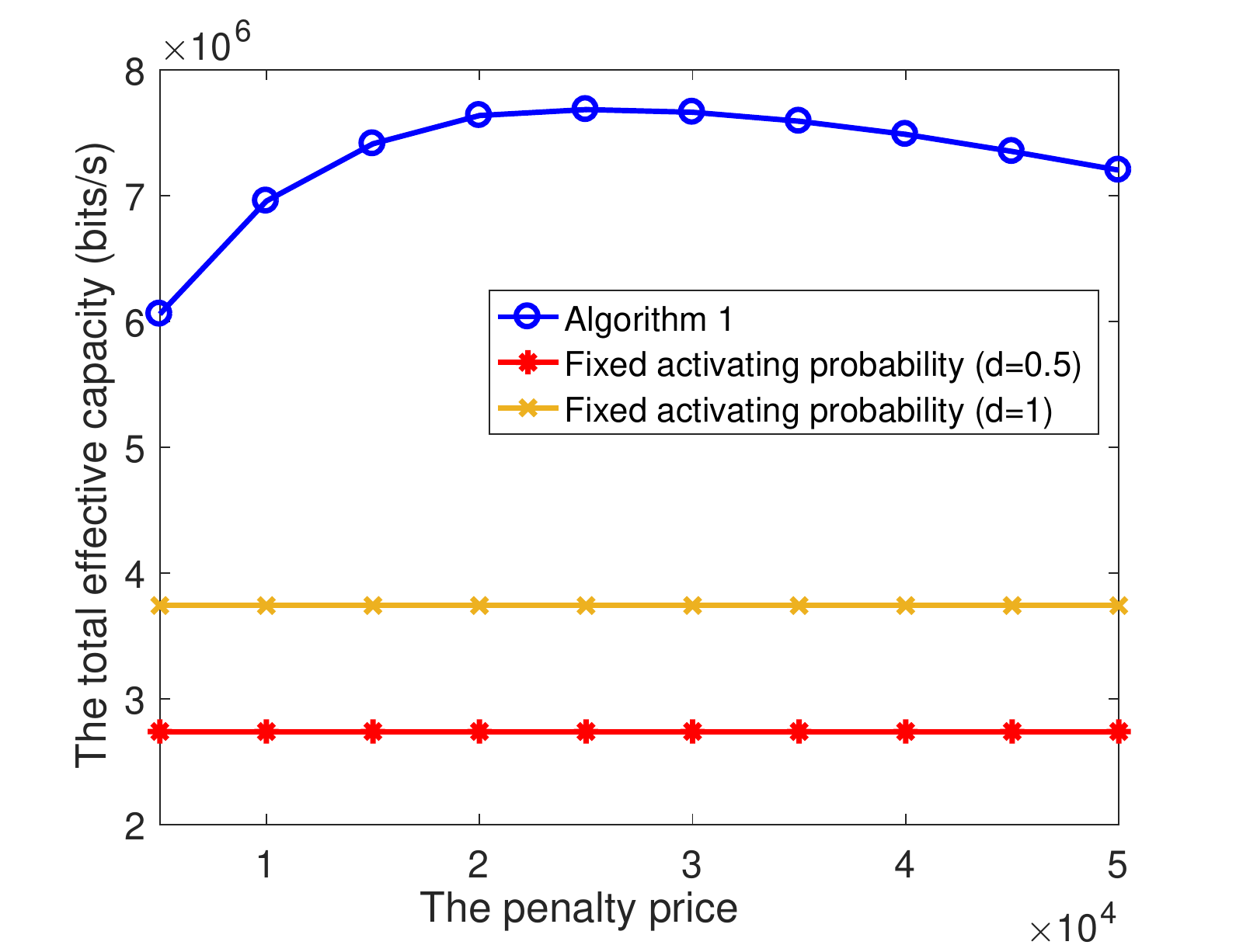}
\vspace{-0.4cm}
\caption{Performance of different strategies with different price.}
\label{fig6}
\vspace{-0.2cm}
\end{figure}

\subsection{Convergence and performance of the price-update barring policy}
In Fig. \ref{fig8}, we investigate the convergence of Algorithm 2, where ${S_{n,k}} = 1000$, and $M=50$.  We observe that Algorithm 2 converges fast to the local optimal point. Fig. \ref{fig9} shows the performance of different algorithms. Compared with the centralized particle swarm optimization (PSO) algorithm, we can find the gain gap is slight between the centralized algorithm and our proposed distributed algorithm. However, our proposed distributed algorithm can fast converge within limited iterations. The price-fixed algorithm 1 cannot leverage device disparity and spatial diversity, resulting in a lower total effective capacity than that of price-update algorithm 2. We also observe that, when the number of iterations is small, the gap between the two algorithms becomes bigger because the preamble resources are so scarce that the lower the price becomes, the severer the collisions become. An appropriate price will alleviate the burst arrivals.

\begin{figure}[!t]
\centering
\includegraphics[width=3.2in]{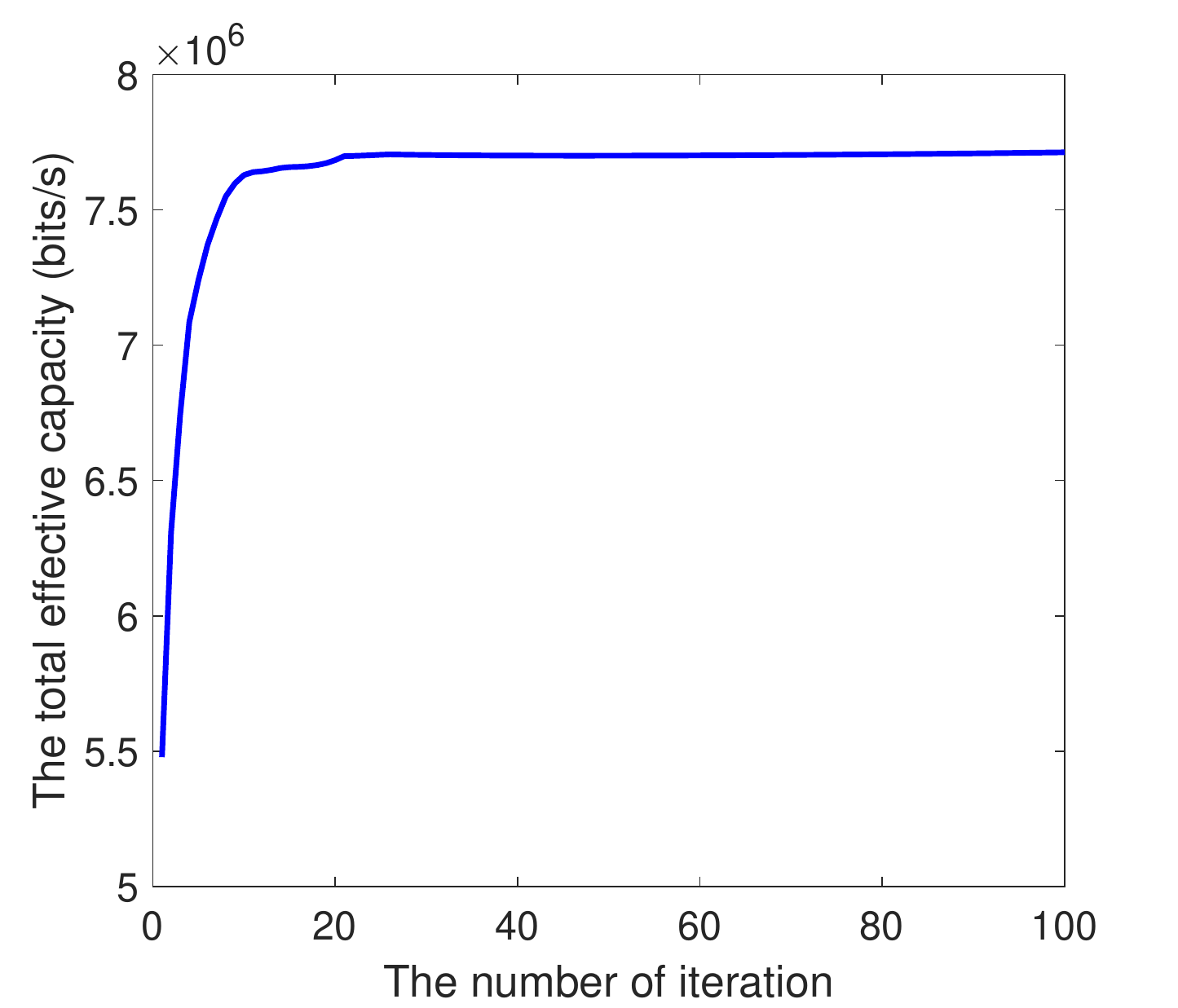}
\vspace{-0.4cm}
\caption{Convergence of Algorithm 2.}
\label{fig8}
\vspace{-0.2cm}
\end{figure}

\begin{figure}[!t]
\centering
\includegraphics[width=3.2in]{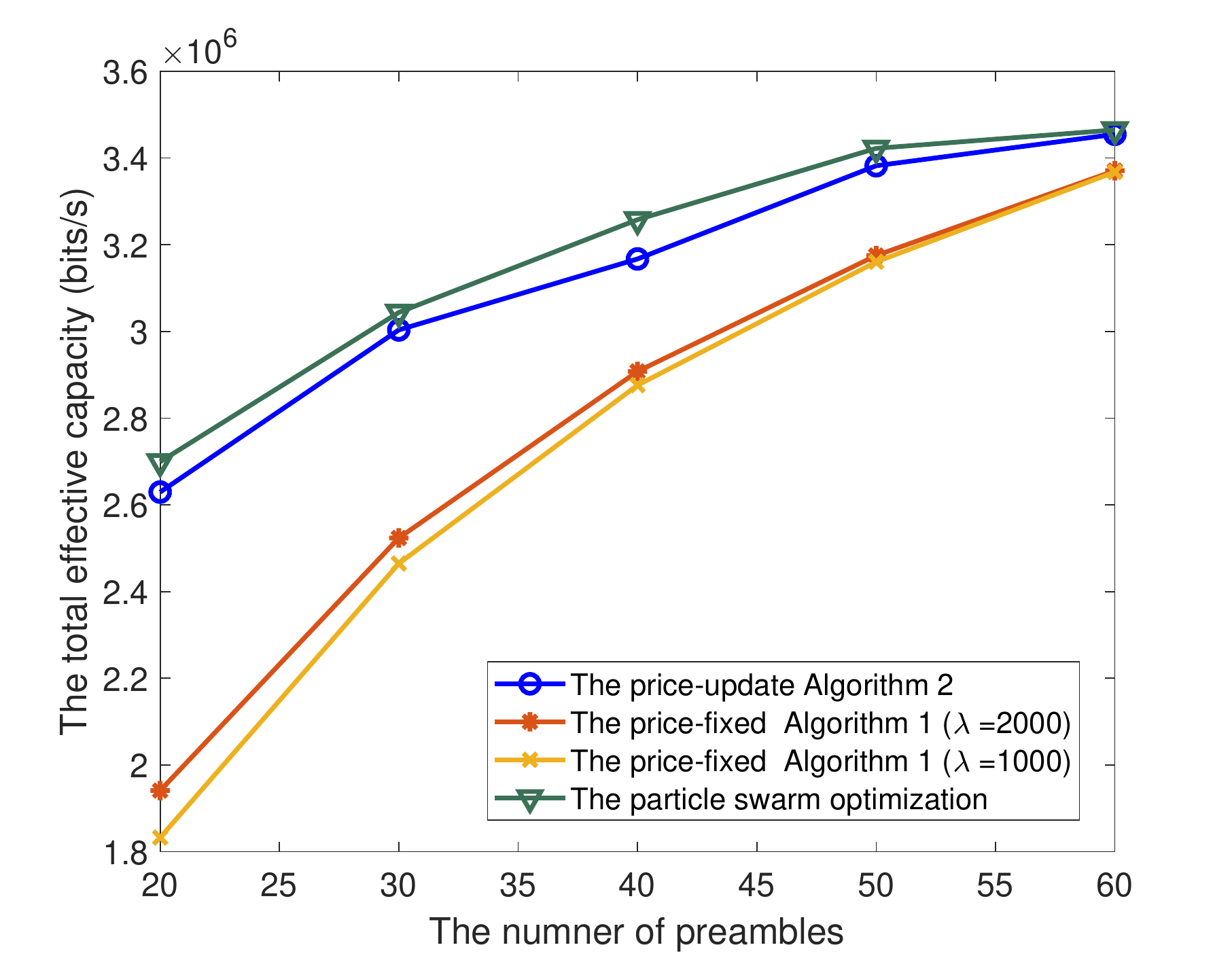}
\vspace{-0.4cm}
\caption{Performance of different algorithms.}
\label{fig9}
\vspace{-0.2cm}
\end{figure}

\section{Conclusion}
In this paper, we investigate a hybrid random access and data transmission protocol design to guarantee the end-to-end QoS of different types of traffic for mMTC. First, we derive a novel effective capacity based on a priority-queueing-based ACB policy. With the proposed effective capacity, the barring policy can be modeled as a non-cooperative game. A distributed iterative algorithm that converges to the NE is then proposed. To further improve the efficiency of NE, we provide a price-update algorithm, which converges to a local optimal point for the effective capacity optimization problem. Simulation results demonstrate the performance of the novel effective capacity, and verify the convergence and the gains of our proposed algorithms.

\section*{Acknowledgment}
The work was supported in part by the National Nature Science Foundation of China Project under Grant 61471058, in part by the Hong Kong, Macao and Taiwan Science and Technology Cooperation Projects under Grant 2016YFE0122900, in part by the Beijing Science and Technology Commission Foundation under Grant 201702005 and in part by the 111 Project of China under Grant B16006, in part by BUPT Excellent Ph.D. Students Foundation (CX2017305) and the China Scholarship Council (CSC).



\end{document}